\documentclass[epj,10pt]{svjour}

\usepackage{mathptmx}
\usepackage[T1]{fontenc}
\usepackage[utf8]{inputenc}
\usepackage{verbatim}
\usepackage{amsmath}
\usepackage{graphicx}
\usepackage{amssymb}
\usepackage{esint}
\usepackage{enumerate}

\newcommand{\lyxdot}{.}

\makeatletter

\usepackage{mathpazo}
\usepackage{mathptmx}

\usepackage{nameref}

\usepackage{longtable}
\usepackage[dvips, ps2pdf, unicode, bookmarks,  plainpages = false]{hyperref}

\usepackage{palatino}
\bibdata
\bibstyle

\makeatother

\begin{document}

\title{Quantum Walk with Jumps}

\author{H. Lavi\v cka\inst{1,2} \and V. Poto\v cek\inst{1} \and T. Kiss 
\inst{3} \and E. Lutz\inst{4} \and I.
Jex\inst{1}}

\institute{Czech Technical University in Prague, Faculty of Nuclear Sciences and Physical Engineering, Department of Physics, Břehová 7, CZ-115 19 Praha 1, Czech Republic, \email{hynek.lavicka@fjfi.cvut.cz} \and
Czech Technical University in Prague, Faculty of Nuclear Sciences and Physical Engineering, Doppler Institute for Mathematical Physics and Applied Mathematics, Břehová 7, CZ-115 19 Praha 1, Czech Republic \and
Research Institute for Solid State Physics and Optics, Hungarian Academy of Sciences, H-1525 Budapest, P.O.Box 49, Hungary \and
Department of Physics, University of Augsburg, D-86135 Augsburg, Germany}

\date{Received: date / Accepted: date}

\date{\today}


\abstract{We analyze a special class of $1$-D quantum walks (QWs) realized
using optical multi-ports. We assume non-perfect multi-ports showing
errors in the connectivity, i.e. with a small probability the multi-ports
can connect not to their nearest neighbor but to another multi-port
at a fixed distance -- we  call this a jump.
 We study two cases of QW with jumps where multiple displacements can emerge 
at one timestep. The first case assumes time-correlated jumps (static 
disorder). In the second case, we choose the positions of jumps randomly in time 
(dynamic disorder).  
The probability distributions of position of the QW walker in both instances differ significantly: dynamic disorder leads to a Gaussian-like distribution, while for static disorder we find two distinct behaviors depending on the parity of jump size. In the case of even-sized jumps, the
distribution exhibits a three-peak profile around the position
of the initial excitation, whereas the probability distribution in
the odd case follows a Laplace-like discrete distribution modulated
by additional (exponential) peaks for long times. Finally, our numerical results indicate that by an appropriate mapping a universal functional behavior of the variance of the long-time probability distribution can be  revealed with respect to the scaled average of jump size.
}

\maketitle

\section{Introduction}

The quantum walk (QW) is a quantum mechanical model,
a generalization of a classical random walk.  It was introduced in 1993 
\cite{aharonov-1993-48,gudder-1988,grossing-1988} and later found fruitful 
applications as a tool to design efficient quantum algorithms.
The model of a quantum walk was defined in two distinct ways: continuous- and 
discrete-time. In the former, the particles (walkers) are achiral, the Hilbert 
space is spanned by the discrete position states corresponding to vertices of 
a graph. In the discrete-time case, the introduction of chirality is 
unavoidable. The Hilbert space corresponding to the chirality has its 
dimension equal to the number of possible directions of a step.

A possible experimental implementation of a classical random walk is the 
Galton board (also known as Quincunx). Here a large number of balls (walkers) 
fall through the board, changing their direction randomly on periodically 
arranged pins and forming so a binomial distribution of their final position.  
A quantum analogy of the Galton board (and one possible implementation of the 
QW) is shown in Fig.~\ref{Material} where the walker is a coherent light pulse 
moving through a medium with periodical boundaries that split the signal;
finally, there are detectors  at the end which represent the quantum equivalent 
of the bins in the classical model.

The spectrum of investigation of the QW (and its modifications) is broad.
The original idea was presented in \cite{aharonov-1993-48,gudder-1988,grossing-1988}
and since then a few review papers have  been published
\cite{kempe-2003-44,meyer96from,Vanegas-Andraca-2008}.
Recently, there have been studies of QW allowing L\'evy noise
in the model. The latter is  introduced via randomly performed measurements with  waiting 
times following a L\'evy distribution 
\cite{romanelli-2007-76,romanelli-2007-75}. 
The most relevant paper for the present work is \cite{yin-2007} in which the properties of 
the one-dimensional continuous-time QW  in 
a medium with static and dynamic disorder are examined. Other studies of the QW 
have focused on the meeting problem of two particles 
\cite{stefanak-2006-39,stefanak-2009-11}.
Recurrence properties of the walker
 have been investigated 
in \cite{stefanak-2009-11,Stefanak-2008,stefanak-2008-78}.
Moreover, localization of the walker has been studied in \cite{konno-2009-3}.
Finally, a theoretical investigation of the QW in random environment has been  performed in 
\cite{konno-2009-4,konno-2009-1}, where a robust mathematical definition of 
a random environment is  provided.

The simplest analytic task in the study of quantum walks is to determine the 
functional expression of the probability to find a particle at a certain 
location at a certain time. General analytic calculations of quantum walks are 
not known. However, asymptotic solutions of several quantum walk models
have been found using path integral techniques in 
\cite{konno-2009-3,konno-2009-4,konno-2009-1,konno-2002,konno-2008-1,konno-2008-2,konno-2009-2,konno-2005-1,konno-2005-2}
and in \cite{stefanak-2006-39,stefanak-2009-11} using Fourier transform.

Since their introduction, many experimental groups have tried to implement quantum walks. A number of successful 
realizations of one-dimensional QW have been reported  in optical 
lattices \cite{Dur-2002}, trapped ions \cite{Travaglione-2002}, and cavity QED 
\cite{Sanders-2003}. More recently, additional realizations of a quantum walk 
using atoms in an optical lattice \cite{Karski-2009}, trapped ions 
\cite{Schmitz-2009,Zahringer-2010} and photons  \cite{Schreiber-2010} have been announced.

Recent studies of quantum systems \cite{ElGhafar-1997,Kim-2000} with random 
potentials have shown that  localization of the particle can occur as 
a result of the randomness. Focusing on quantum walks, Ribeiro {\it et al.}
 \cite{Ribeiro_2004} have used two different coin operators switched according to 
the Fibonacci series and they have observed localization in the system.  Yin {\it  et  
al.}  \cite{yin-2007} have numerically simulated the continuous-time QW on a line and they have observed 
Anderson localization only in the case of static disorder, while dynamic 
disorder leads to decoherence and a Gaussian position distribution.
In addition, effects of spatial errors have been studied by Leung
{\it et al.} in \cite{Leung-2010}. They have examined  QWs in $1$ and $2$ dimensions
on networks with percolation where the missing edges or vertices absorb
the walker, leading to topological randomness of the graph.

The simplest discrete QW is described by the action of the coin and the step operator. Much attention was paid to the alternation of the coin operator -- position or time dependent coin and its implications on the walker dynamics have been extensively discussed. However, little focus was given to changes of the step operator. Our analysis takes a step in this direction. While assuming a constant coin we study changes in the step operator. When assuming a Galton board realization this amounts to a change in the connectivity between the layers of beam splitters forming the walks.

In the present paper, we focus on the QWs where signals can jump to a distant location, which is a generalization motivated by Lévy flights in classical mechanics. The jump may be caused by an inhomogeneity of the material, spatial proximity of non-neighbor channels or scrambling in the topology of the network (for instance relabeling of input-output label of multiports forming the network realising the QW). 

We study two basic modifications of the QW. In the first case, which 
we call dynamic disorder, the jumps are prepared as independent and identically 
distributed in time, whereas in the second case, called static disorder, the 
positions of  the jumps are perfectly correlated. We investigate the problem using 
computer simulations employing the  Zarja  library \cite{Lavicka-2010} 
\footnote{\url{http://sourceforge.net/projects/zarja/}}
and its offspring library focused on the QW 
\footnote{\url{http://sourceforge.net/projects/quantumwalk/}}.  

The structure of the article is as follows. In Section~2, 
we motivate our study from an experimental point of view. Then we define a model with 
next neighbor interactions only. In Section~3, we show the results of the 
simulations based on Monte Carlo method. In Section~4, we discuss our results 
and draw conclusions. Finally,  we describe in the appendix the 
set of operators used in the simulations and an algorithm to sample from the 
appropriate probability distribution.

\section{Definition of a Quantum Walk}

We model a sequence of optical layers by an array of multi-ports, or beam
splitters (Fig.~\ref{Beam splitter}), forming so a large interferometer, see 
Fig.~\ref{Material} -- this is a quantum analog to Galton's
board (Quincunx). Due to the regular structure of the interferometer,
it is natural to treat the temporal evolution of the system in discrete
steps, discretized by the time needed for an excitation to travel the
distance between two consecutive layers at a constant angle. If we
let this angle approach the right angle, we obtain so-called static
disorder (see Fig.~\ref{Jump}d,e), where the same set of beam splitters
is used repeatedly in every time step.
An excitation enters the system of multi-ports at one selected
position and spreads due to the coherent interaction with the beam
splitters. A number of detectors is placed in the interferometer so
that the signal hits one detector in every possible path after a given
number of interactions with the medium. The resulting scheme is an implementation of the QW.

\begin{figure}
\begin{centering}
\includegraphics[scale=0.6]{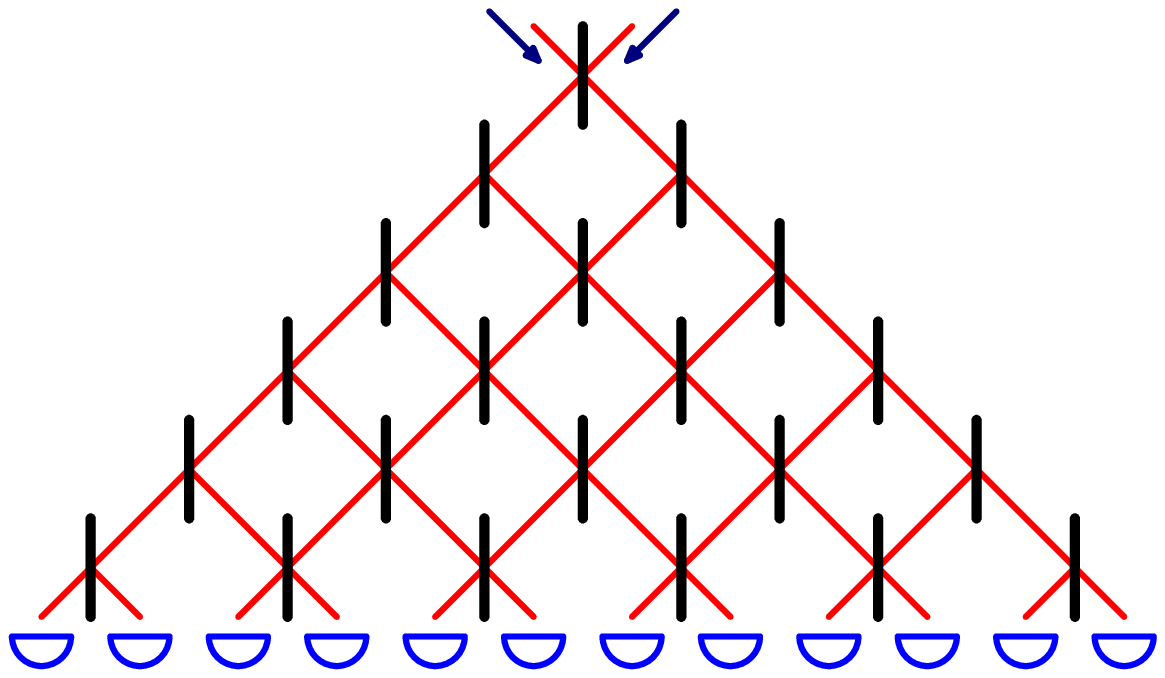}
\par\end{centering}

\caption{A schematic of an interferometer simulating material and implementing
a random walk. The red lines are the connections of multi-ports along
which the signal is transmitted. The light blue half-circles represent 
detectors.
The blue arrows are two selected input channels. Every horizontal
layer of the multi-ports is accessed simultaneously.
}

\label{Material} 
\end{figure}

\begin{figure}
\begin{centering}
\includegraphics[scale=0.15]{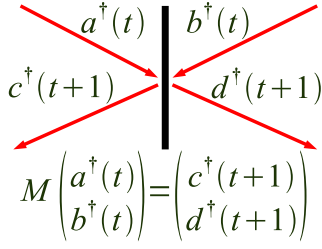} 
\par\end{centering}

\caption{Action of a beam splitter on the input states $a^\dagger (t)$
and $b^\dagger (t)$ which are transformed to the output states
$c^\dagger (t+1)$ and $d^\dagger (t+1)$ using the matrix $M$.}

\label{Beam splitter} 
\end{figure}

 In the following, we introduce two ways to describe the system. First we introduce the standard notation for a discrete-time QW.
We define the basis
states of the system as states localized at the beam splitters just
before or after the interaction takes place. Therefore, every basis state is described by a position within the beam
splitter layer (an integer number) and a chirality, which can take
values $L$ or $R$ for left and right, respectively.
Formally, we define the state space $\mathcal{H}$ of the system as
\begin{equation}
\mathcal{H}=\mathcal{H}_{S}\otimes\mathcal{H}_{C}\end{equation}
where a {}``position'' space $\mathcal{H}_{S}$ and a {}``coin''
space $\mathcal{H}_{C}$ (the name stemming from the idea of a random
walker tossing a coin to decide the direction of his next step) are
defined as
\begin{equation}
\mathcal{H}_{S}=\text{Span}\lbrace|n\rangle\mid n\in\mathbb{Z}\rbrace
\end{equation}
and
\begin{equation}
\mathcal{H}_{C}=\text{Span}\lbrace|L\rangle,|R\rangle\rbrace.
\end{equation}
We define all the numbered states $|n\rangle$, as well as the states
$|L\rangle$ and $|R\rangle$, to be normalized and mutually orthogonal.
Therefore, $\mathcal{H}_{S}$ is isomorphic to the space $\ell^{2}$ of 
quadratically integrable complex sequences and $\mathcal{H}_C$ to 
$\mathbb{C}^{2}$. The basis states are then constructed as tensor product 
basis states $|n,R\rangle$ and $|n,L\rangle$, where $n\in\mathbb{Z}$.

Following Konno \cite{konno-2009-1}, we define
that the state of the system at time $t$ in a random environment
$\omega$ 
 is described by the positive semi-definite density 
matrix $\widehat{\rho}_{\omega}(t)$ on $\mathcal{H}$
having $\mathbf{Tr}\ \widehat{\rho}_{\omega}(t)=1$. Thus
the density matrix of the system $\widehat{\rho}(t)$ at time $t$
is \begin{equation}
\widehat{\rho}(t)=\int_{\omega\in\Omega}\widehat{\rho}_{\omega}(t) \cdot 
\mathbb{P}(d\omega),\label{Definition of density matrix of 
system}\end{equation}
where symbols $\Omega$ and $\mathbb{P}$ will be specified later in section 2.2.

The evolution of the system by one timestep from $t$ to $t+1$ for
one particular setup of environment $\omega$ is described by a unitary
operator $U_{\omega}(t\rightarrow t+1)$ which acts as follows: \begin{equation}
\widehat{\rho}_{\omega}(t+1)=U_{\omega}(t\rightarrow t+1)\cdot\widehat{\rho}_{\omega}(t)\cdot U_{\omega}^{\dag}(t\rightarrow t+1)\label{Evolution equation of QW}\end{equation}
just the same way as typical propagator on density matrix in quantum
mechanics in, e.g., \cite{Blank_1993}.

\subsection{Quantum Walk}

The QW described, e.g., in \cite{kempe-2003-44} assumes that the conditions
of the environment are stable and thus it holds 
$\widehat{\rho}(t)=\widehat{\rho}_{\omega}(t)$
for all $\omega\in\Omega$. If we set up the initial density matrix
$\widehat{\rho}(0)$ as a pure state $\vert\psi_{0}\rangle\in{\cal H}$
using $\widehat{\rho}(0)=\vert\psi_{0}\rangle\langle\psi_{0}\vert$,
then the system is in a pure state $\vert\psi_{T}\rangle=\prod_{t=0}^{T-1}$ $U^{QW}(t\rightarrow t+1)\vert\psi_{0}\rangle$
(the product expanded in the proper time ordering) at every timestep
$T$ and the appropriate density matrix is $\widehat{\rho}(T)=\vert\psi_{T}\rangle\langle\psi_{T}\vert$.
Moreover, the  original QW was defined with fixed unitary operator
$U^{QW}(t\rightarrow t+1)$ in the form of a composition of two unitary
operations as follows: \begin{equation}
U^{QW}(t\rightarrow t+1)=S\cdot C.\label{QW unitary operator}\end{equation}
 Thus, the evolution of the walker at time $T$ is $\vert\psi_{T}\rangle =$
$\left(U^{QW}\right)^{T}\vert\psi_{0}\rangle$.
The operator $C$ is called a \textsl{coin operator} and describes
the transformation induced by a beam splitter. Here the position of
a localized state stays unchanged and the operation acts only on the
coin state as $C=I\otimes M$, where $M$ is a unitary operation transforming
the probability amplitudes due to a partial reflection on the beam splitter.
We assume for simplicity that all the beam splitters have the same
physical properties and perform a Hadamard transform on the input
states, \begin{equation}
M=\frac{1}{\sqrt{2}}\left(\begin{matrix}1 & \phantom{-}1\\
1 & -1\end{matrix}\right).\end{equation}
The operator $S$ represents the propagation of an excitation in the
free space between the beam splitters. Hence, the chirality of a basis
state does not change but the position is shifted by $\pm1$, depending on
the coin state. We can express $S$ as \begin{equation}
S=\sum_{n\in\mathbb{Z}}(\vert n+1,R\rangle\langle n,L\vert+\vert 
n-1,L\rangle\langle n,R\vert).\end{equation}
 It is easy to show that the sum converges and defines a unitary operator 
defined on all the state space $\mathcal{H}$.

If the initial state of the walker is one of the basis states, for
example $|0,R\rangle$, it evolves under the operation $U$ so that
in terms of a complete measurement in the position space, the probability
spreads to both sides from the starting position. However, it is bounded
between the positions $-t$ and $t$ as it can't change by more than
$1$ in either direction in any step. Moreover, due to the fact that
a transition by $1$ has to be done in every time step, the walk is
restricted at each time $t$ to a subspace spanned by the basis states
for which the position shares the same parity with $t$.

\global\long\def\half{\ensuremath{{\textstyle \frac{1}{2}}}}

Alternatively, we can describe the path of the walk in terms of the edges of 
the underlying graph instead of its vertices, following the physical 
trajectories of the excitation and eliminating the need to describe the 
propagation between beam-splitters.  In the following, we will call these 
edges {}``channels''. Due to the fact that every channel lies in between two consecutive 
positions of the beam splitters, as well as to avoid confusion in the 
notation, we will denote the channels by
half-integer numbers.
A notation based on the channel formalism can be introduced and mapped
to the previously defined state space in more ways. One such possibility
is to identify the channel as a given time $t$ with the state in
which the excitation is at the \textsl{end} of the propagation, just
before hitting another beam splitter (or a detector). This time-dependent
mapping is given by the formulas \begin{equation}
|n+\half;t\rangle=\begin{cases}
|n,L\rangle & \text{if \ensuremath{t} is even},\\
|n+1,R\rangle & \text{if \ensuremath{t} is 
odd}\end{cases}\label{Mapping_of_states_1}\end{equation}
 for even $n$ and \begin{equation}
|n+\half;t\rangle=\begin{cases}
|n+1,R\rangle & \text{if \ensuremath{t} is even},\\
|n,L\rangle & \text{if \ensuremath{t} is 
odd}\end{cases}\label{Mapping_of_states_2}\end{equation}
 for odd $n$. One can verify that such a mapping respects the parity
rule and covers all the subspace of $\mathcal{H}$ that is actually
used by the quantum walk if we assume that the initial state was $|0,R\rangle$
or $|0,L\rangle$. In particular, these two initial states are denoted
$|-\half;0\rangle$ and $|\half;0\rangle$ in the channel
notation, respectively; cf. Fig.~\ref{Material} for a visualization.
Since the instantaneous chirality of an excitation is uniquely given
by the position of the channel and the time, we do not need to specify
the coin degree of freedom explicitly in this approach. The {}``coin
toss'' and {}``step'' are merged into one unitary operation which
mixes neighboring channels in pairs.

The above discussion gives two equivalent ways to describe QW on a line. 
Throughout this work, we will use the latter approach as it makes it much simpler to describe the jumps in the network.

\subsection{Quantum Walk with jumps}

We will assume that the quantum walk is disturbed by random topological 
errors, see Fig.~\ref{Jump}, which are modeled by random changes of 
connectivities between the multi-ports
forming the network in Fig.~\ref{Material}. We focus our analysis
on two distinct situations. First we will assume that we deal with
the repetition of random but stationary errors. 
Stationary errors mean that in each layer the same jumps appear. One could also realize static disorder with a single set of beam splitters, if the incident beam is parallel to the layer of beam splitters and repeatedly sent through them (see Fig.~\ref{Jump}d,e). The other situation refers to the case when in each layer jumps are generated independently. We refer to the latter situation as dynamic disorder (see Fig.~\ref{Jump}a,b,c).

In the QW with jumps we assume that the unitary operators 
$U_{\omega}(t\rightarrow t+1) = U_{\omega}^{jump} (t\rightarrow t+1)$ in Eq. \ref{Evolution equation of QW}
are random (the set of such operators forming a probability space of
random unitary operators) and that $U_{\omega}^{jump}(t\rightarrow t+1)$
can be written in form of a joint action of two unitary operations,
\begin{equation}
U_{\omega}^{jump}(t\rightarrow t+1)=S_{\omega}^{jump}(t)\cdot U_{\omega}^{QW}(t\rightarrow t+1), \label{QW with jumps unitary operator}
\end{equation}
where $U_{\omega}^{QW}(t\rightarrow t+1)$ is the evolution operator of the $1$D QW in a clean media without errors, defined by Eq. \ref{QW unitary operator}.

We define that the set of all unitary operators $S_{\omega}^{jump}(t)$
forms the probability space $S(\Omega,{\cal F},\mathbb{P})$. For the sake of simplicity
of the following discussion and the numerical simulations, we make the set
$\Omega$ finite by replacing the infinite walking space by a cycle of size $N$
with a periodic boundary condition. If $N$ is chosen sufficiently large, this 
imposes no restriction on the validity of the results. The sample space
$\Omega$ is then the set of all the possible combinations of jump operators
$\mathbf{{P}}_{j_{i}}$ exchanging signals in channels $j_{i}$ and
$j_{i}+j$, having the matrix form \begin{equation}
\mathbf{{P}}_{j_{i}}=\left(\begin{smallmatrix}1 & 0 & \ldots & 0 & \ldots & 0 & \dots & 0 & 0\\
0 & 1 & \ldots & 0 & \ldots & 0 & \dots & 0 & 0\\
\vdots &  & \ddots &  &  &  &  &  & \vdots\\
0 & 0 & \ldots & 0 & \ldots & 1 & \dots & 0 & 0\\
\vdots &  &  &  & \ddots &  &  &  & \vdots\\
0 & 0 & \ldots & 1 & \ldots & 0 & \dots & 0 & 0\\
\vdots &  &  &  &  &  & \ddots &  & \vdots\\
0 & 0 & \ldots & 0 & \ldots & 0 & \dots & 1 & 0\\
0 & 0 & \ldots & 0 & \ldots & 0 & \dots & 0 & 1\end{smallmatrix}\right),\label{Structure_of_propagation_operator}
\end{equation}
i.e., that of a transposition operator. ${\cal F}=2^{\Omega}$
is the $\sigma$-field defined on
\begin{equation}
\Omega=\left\lbrace 
\mathbf{E},\mathbf{P}_{j_{1}},\mathbf{P}_{j_{1}}\mathbf{P}_{j_{2}},\ldots
\right\rbrace
\end{equation}
and $\mathbb{P}:{\cal F}\rightarrow\left[0,1\right]$
is the probability measure, specified by the probability of elementary events, 
defined as
\begin{equation}
\mathbb{P}\left(\pi\right)=\frac{1}{Z}p^{tr(\pi)}\left(1-p\right)^{N-2\cdot 
tr(\pi)}.
\end{equation}
We define $tr(\pi)$ as the number of transpositions of indexes forming
permutation $\pi$, $N$ is equal to the size of the system and consequently to the dimension of jump operators $\mathbf{{P}}_{j_{i}}$,
$p$ is probability that one pair of errors with distance $j$ occurs
and finally $Z=\left( 1+\left( - p \right)^{\frac{N}{g}}\right)^g$ is 
normalization where $g= \mathrm{gcd} (N,j)$ ($\mathrm{gcd}$ stands for greatest common denominator). 

\begin{figure}
\begin{centering}
\includegraphics[scale=0.7]{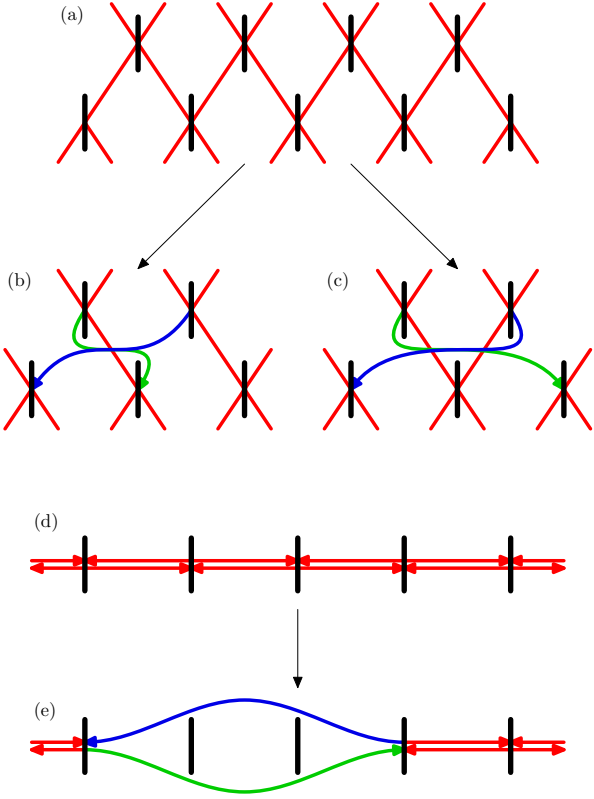} 
\par\end{centering}

\caption{Schematic picture of different jumps of the signal. 
Scheme (a) shows a part of interferometer without jumps.
Two schemes in the middle, namely (b) and (c), show dynamic disorder pattern where (b) and (c) shows even ($j=2$) and odd ($j=3$) size of jump, respectively, which can emerge independently at every level. Scheme (d) shows an interferometer without jumps for static disorder and a jump of size $j=2$ is illustrated in scheme (e).}

\label{Jump} 
\end{figure}

\section{Simulation of the model}

The disorder of the system is expressed by the set of unitary operators
displacing the signal to a distant location. In general, the system
can undergo either static or dynamic disorder. The system with static
disorder is propagated by  random unitary operators fixed during the realization,
i.e., the operators are correlated in time. On the other hand, the evolution
of the system undergoing dynamic disorder is governed by the operators
that are independently and identically distributed in time.

\subsection{Static disorder}

The random unitary operators propagating the system with typical structure of 
jumps according to Eq.  \ref{Structure_of_propagation_operator} can have due 
to mapping \ref{Mapping_of_states_1} and \ref{Mapping_of_states_2} one of two 
fundamental distinguished forms, according to the parity of the length of the 
jumps induces by the errors in the media.
Jumps of even lengths $j$ do not change chiral state of the walker in Hilbert 
state ${\cal H}$ in contrast to jumps of odd lengths $j$ which swap chirality 
of the walker from state $L$ to $R$ and vice versa.

\subsubsection{Odd jumps}

The only tunable parameter that affects the evolution of the walker is
the probability $p$ that one error (jump with distance $j$) occurs. The value of the 
parameter $p$ close to $0$ should produce typical chiral $1$D QW pattern,
which is clearly visible on the top of 
Fig.~\ref{Quantum_random_walk_odd_jumps}. Here the QW was initiated in 
a localized state $|0+\half;0\rangle$ of the walker,
causing an asymmetrically distributed walk. Higher values of $p$ produce
totally different pattern where high-frequency oscillations of probability
distribution of positions of the walker are suppressed. In the middle and at 
the bottom of Fig.~\ref{Quantum_random_walk_odd_jumps} the probability
distribution of the positions of the walker shows Laplace distribution
modulated by Laplace distributed peaks with distance $j$ between
neighboring peaks (clearly observable as the small triangular peaks modulated
on a bigger structure in Fig.~\ref{Quantum_random_walk_odd_jumps}).

\begin{figure}

\begin{centering}
{\includegraphics[scale=0.67]{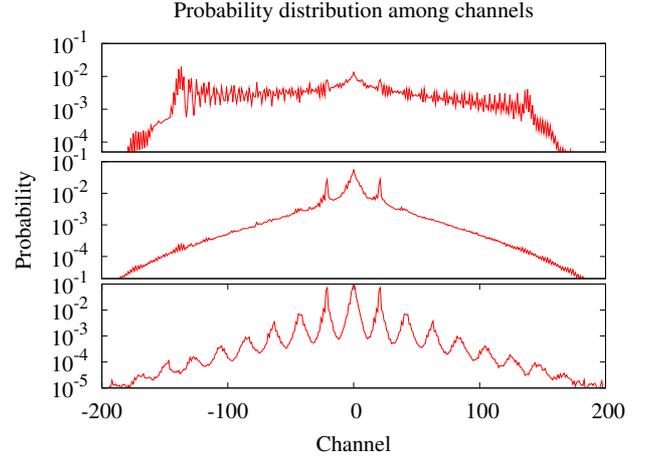}}

\par\end{centering}

\caption{Typical patterns of probability distribution of position of the walker
among channels after $T=200$ steps for $j=21$, from the top to the bottom $p=0.01$, $p=0.05$ and $p=0.2$.
On the top, the typical structure of the QW without jumps that is perturbed by formation
of an extra central peak. In the middle, we see formation of Laplace-like tail and $3$ peaks with distance $j$. Central peak is located at the position of the initial excitation of the system at $T=0$. On
the bottom, the whole distribution follows an overall exponential
decay which is modulated by exponential peaks with a distance $j=21$
between maxima. The graphs were obtained using Monte Carlo method after $R 
= 20000$ runs of the randomized evolution.}

\label{Quantum_random_walk_odd_jumps} 
\end{figure}

Previous observations of the fundamental change induced by  the variation of $p$ is
supported by Fig.~\ref{Quantum_random_walk_odd_jumps_distribution}
showing evolution of the probability distribution. The top part shows the 
typical structure of evolution of the probability distribution of the walker---the quantum carpet \cite{Berry-2001}---with additional quantum carpets
on the border of the main one which were formed at initial stages of 
evolution.
For small values of $p$, the interference is not strong enough to change the pattern
of distribution. In contrast, in the second case displayed at the bottom of 
Fig.~\ref{Quantum_random_walk_odd_jumps_distribution}, where $p = 0.5$, we 
observe a typical structure of equidistant peaks separated by valleys of width 
$j$.
The peaks are formed early during evolution and they do not change their 
positions later.

\begin{figure}
\begin{centering}
{\includegraphics[scale=0.76]{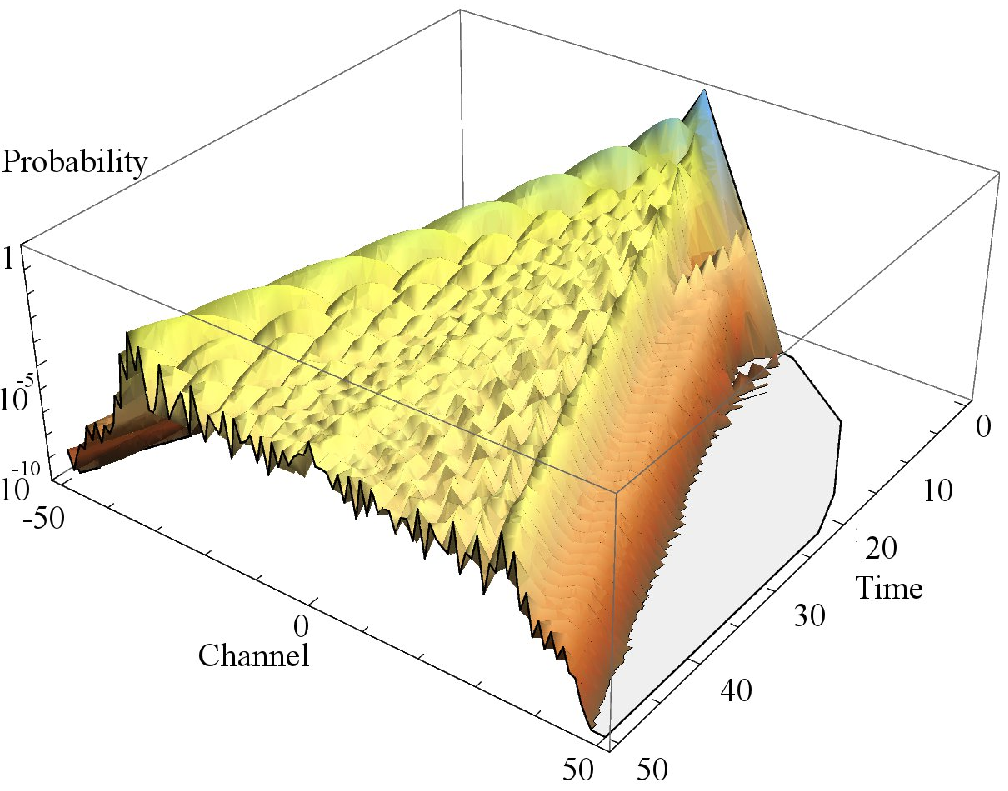}}
{\includegraphics[scale=0.76]{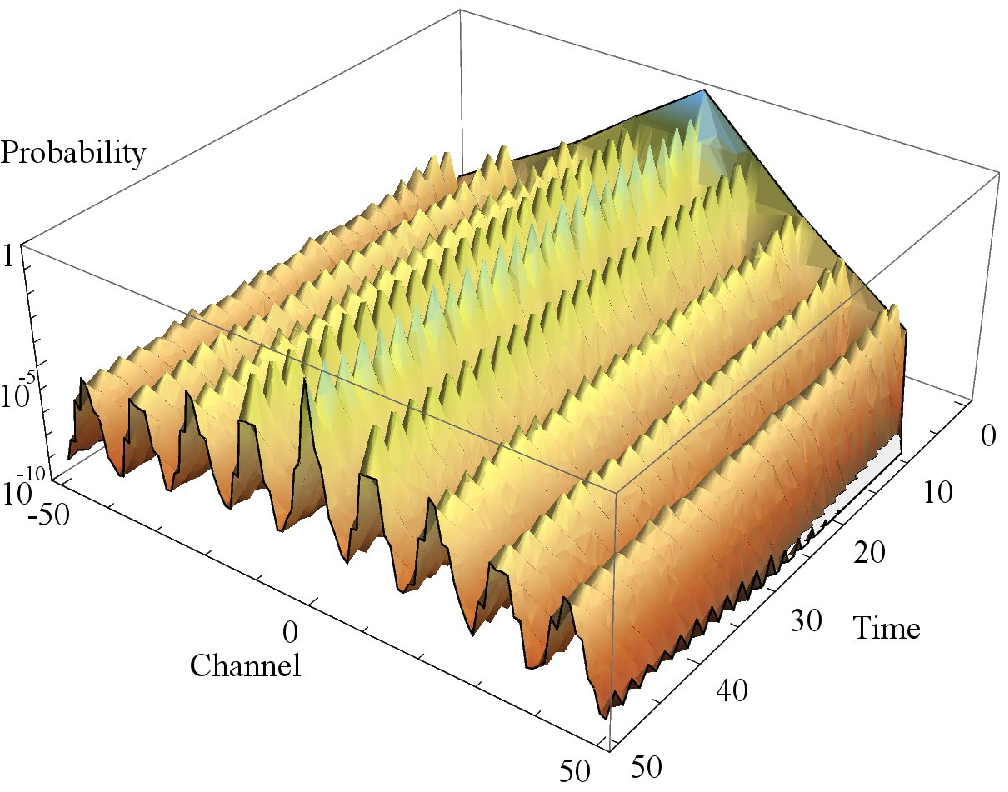}} 
\par\end{centering}

\caption{Probability distribution among channels for $T=200$, $R=20000$ and
$j=11$; from the top to the bottom $p=0.01$ and $p=0.5$. On the
top, the maxima of probability move outward similarly to the standard
case of the QW. On the bottom, the maxima are formed and they stay
at the point of original formation. }

\label{Quantum_random_walk_odd_jumps_distribution} 
\end{figure}

To support the idea that the walker froze in the close neighborhood
of a few preferred locations we analyzed the evolution of entropy and variance
of the probability distribution of the position of the walker. The evolution of 
the variance, shown in Fig.~\ref{Evolution_of_variance}, visualizes the 
observation from the previous paragraph. For a probability $p$ close to $0$ 
(unperturbed QW), we see a clear ballistic diffusion of the walker ($\sigma^2 
\sim t^{2}$)---in contrast to this we stress that for a classical random 
walker $\sigma^2 \sim t$.  Increasing $p$, we observe sub-ballistic diffusion 
and finally, in the range close to $1/2$, the variance tends to finite 
constant.

\begin{figure}
\begin{centering}
{\includegraphics[scale=0.6]{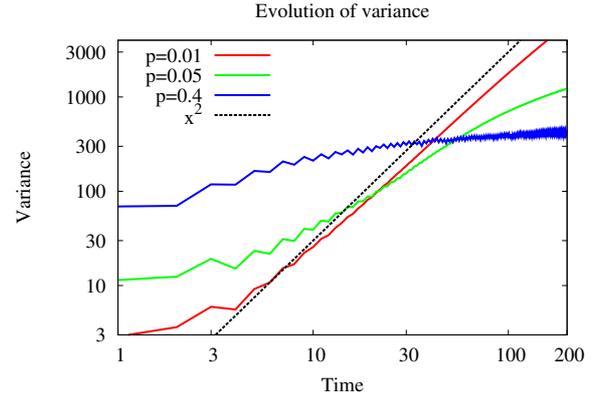}} 
\par\end{centering}

\caption{Evolution of variance of the system f or $T=200$, $R=20000$ and 
$j=11$. Transition from the case of ballistic diffusion of the signal
for probabilities $p$ close to $0$ to the case of sub-ballistic diffusion,
demonstrated for $p=0.05$ can be observed, and finally we observe for $p=0.4$ that variance 
tends to a constant, i.e., diffusion coefficient approaches $0$ and the walker 
ceases to {}``spread''---a characteristic property of localization 
\cite{anderson-1958}.
}

\label{Evolution_of_variance} 
\end{figure}

Next, we turn our attention to the evolution of the classical (Shannon) entropy
of the probability distribution of the positions of the walker, which is shown
in Fig.~\ref{Evolution_of_entropy}. The evolution of the classical entropy
follows an analogous behaviour as described in the case of variance in the 
previous paragraph.  Thus, we can observe for the probability $p$ close to 
unperturbed QW an increase of classical entropy proportional to 
$\ln\left(t\right)$. However, rising probability $p$ causes slower and slower 
increase in entropy in time and for the case $p=0.4$ we observe a saturation 
of entropy.

\begin{figure}

\begin{centering}
{\includegraphics[scale=0.6]{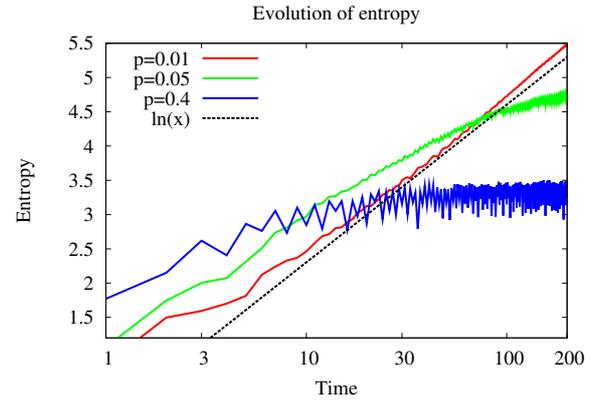}} 
\par\end{centering}

\caption{Evolution of classical (Shannon) entropy of probability distribution
of positions of walker for $T=200$, $R=20000$ and $j=11$. Classical entropy is a rising function of time $t$ for $p=0.01$
where the functional form is proportional to $\ln t$. On the other
hand, increasing probability $p$ causes {}``freezing'' of the evolution, as 
can be observed for $p=0.4$.}

\label{Evolution_of_entropy} 
\end{figure}

The conclusion of the above paragraphs on odd-sized jumps in the QW is
that we clearly observe localization of quantum walker in
static disorder media described by the probability space of random unitary
operators $S(\Omega,{\cal F},P)$ parametrized by probability $p$.
Moreover, the numerical results suggest that taking separately odd and even 
timesteps the probability distribution of the position of the walker converges 
to a stationary distribution: \begin{equation}
\langle{n+\frac{1}{2};2t}\vert\widehat{\rho}(2t)\vert{n+\frac{1}{2};2t}\rangle \xrightarrow{t\rightarrow+\infty} \mathbb{P}_{p,j}^{even}\left(n+\frac{1}{2}\right)\end{equation}
 \begin{eqnarray}
\langle{n+\frac{1}{2};2t+1}\vert\widehat{\rho}(2t+1)\vert{n+\frac{1}{2};2t+1}\rangle\xrightarrow{t\rightarrow+\infty} \nonumber 
\\
\mathbb{P}_{p,j}^{odd}\left(n+\frac{1}{2}\right).\end{eqnarray}
  $\mathbb{P}_{p,j}^{even}(n$ 
${}+\frac{1}{2})$
and $\mathbb{P}_{p,j}^{odd}(n+\frac{1}{2})$ are universal distributions
for odd and even timesteps for probabilities $p$ in range close to
$\frac{1}{2}$ and the size of the region of convergence is also dependent
on $j$.

Let us consider the fundamental properties
of the asymptotic probability distribution 
$\mathbb{P}_{p,j}^{even}\left(n+\frac{1}{2}\right)$; variables of the 
probability distribution $\mathbb{P}_{p,j}^{odd}\left(n+\frac{1}{2}\right)$
for odd timesteps are different but general properties are shared.
We concluded from Fig.~\ref{Quantum_random_walk_odd_jumps} that we
observe a formation of an overall Laplace distribution modulated by Laplace
distribution of peaks, both in the form \begin{equation}
\mathbb{P}\left(x\right)=C\exp\left(-\frac{\vert 
x-\mu\vert}{a}\right),\label{Laplace_distribution}\end{equation}
 where $\mu$ is the mean value of the distribution and $a$ is related to 
variance via
$Var\ \mathbb{P}\left(x\right)= 2\cdot a^2$. Our aim is to estimate the inverse parameter $\frac{1}{a}$ of Laplace
distribution in two cases for 
\begin{itemize}
\item the whole distribution (with mean at the point of injection), 
\item the modulated peaks (with mean at the center of the peak). 
\end{itemize}

We focus on the shape of the whole probability distribution. The plots
in Fig.~\ref{Quantum_random_walk_odd_fits_whole} suggest a U-shape
function of the fitted inverse parameter $1/a$ of the Laplace 
distribution for all parameter values. We vary the probability
$p$ and we connect points for the same jump sizes $j$. Defining the $x$-axis 
as $x=p\cdot j$ we put all minima at a fixed position. Thus,
for the minima  $p \cdot j=x_{min}$  holds
where an approximation of the constant 
is $x_{min}=2$. The values of minima of the fitted inverse parameter 
$1/a$, reached for $p=x_{min}/j$, form an increasing function 
of jump size $j$.
\begin{figure}
\begin{centering}
{\includegraphics[scale=0.66]{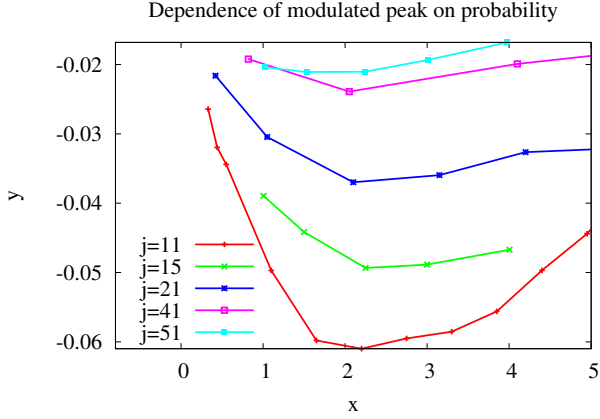}}

\par\end{centering}

\caption{Fit of the inverse parameter $\frac{1}{a}$ of the probability distribution
of position of the walker for $T=200$ and $R=15000$ with various odd sizes
of jump $j$ taking $x$-axis as $x=p\cdot j$ and $y=1/a$. We observe
the formation of a U-shape function for constant sizes of jump $j$ and 
changing $x=p\cdot j$ with a minimum at $x_{min}=2$.
}
\label{Quantum_random_walk_odd_fits_whole} 
\end{figure}
The size of the central peak is $j$, $j/2$ on the left and
on the right from the maximum at $0$. Estimation of the inverse parameter
$1/a$ of the central peak is shown in Fig.~\ref{Quantum_random_walk_odd_fits}.
The shape of the central peak is an increasing linear function of probability
of jump $p$ independent of the size of jump $j$. Thus the central
peak becomes steeper and steeper with increasing $p$---localization
of the walker becomes more evident. This is  true not only for
the central peak but it holds true for the other peaks as well, see Fig.~\ref{Quantum_random_walk_odd_jumps}.

\begin{figure}

\begin{centering}
{\includegraphics[scale=0.66]{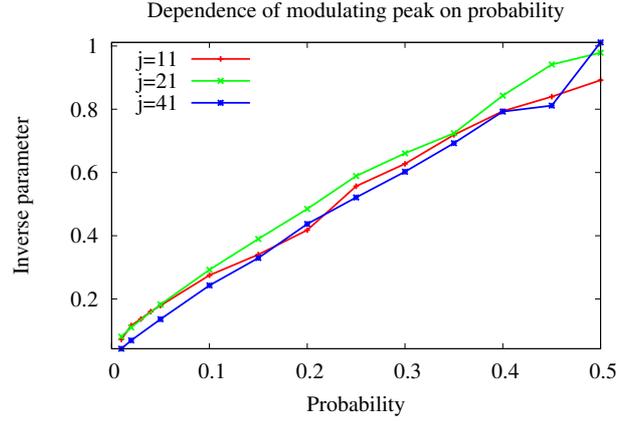}}

\par\end{centering}

\caption{Fit of the inverse parameter $\frac{1}{a}$ of the central peak of
the probability distribution of the position of the walker for $T=200$
and $R=15000$ with various odd sizes of jump $j$.
Increasing probability causes linear increase of the inverse parameter
$\frac{1}{a}$ with $p$ when $j$ is fixed.
Moreover we observe a universal dependence on $p$, independent of the length
of jump $j$.
}

\label{Quantum_random_walk_odd_fits} 
\end{figure}

Let us look at the dependence of entropy of probability distribution $\mathbb{P}_{j,p}\left(n+\frac{1}{2}\right)$
to find a walker at position $n+\frac{1}{2}$. We plot it in 
Fig.~\ref{Quantum_random_walk_odd_entropy} where the entropy was measured by 
both extensive and non-extensive measures. In the first case, the extensive 
measure is classical (Shannon) entropy. The non-extensive one is the 
$q$-entropy introduced by Tsallis, which for a particular $q=1$ reduces to the classical entropy (for more about $q$-statistics, see 
\cite{Tsallis-1988,Tsallis-2004,Tsallis-2009a,Tsallis-2009b,Tsallis-2009c}).  
Both entropies
are decreasing functions for increasing disorder measured by $p$.
This leads to a counter-intuitive statement
 that increasing classical
disorder organizes quantum system even if measured by a non-extensive entropy (in classical random walk longer jump causes increasing
variance and increasing entropy as well).
Moreover, the $q$-entropy of the probability distribution of position
of the walker for the parameter value $q=2$ brings all curves corresponding to 
different sizes of jump $j$ on one single curve. This means that if we take into 
account non-additivity of the system quantified by the parameter 
$\vert1-q\vert$ (taken form the theory of the nonadditive $q$-entropy), we can map the QWs between each other for different sizes of 
error $j$ holding probability of occurrence of pair of error positions $p$ 
constant for $q=2$.

\begin{figure}

\begin{centering}
{\includegraphics[scale=0.66]{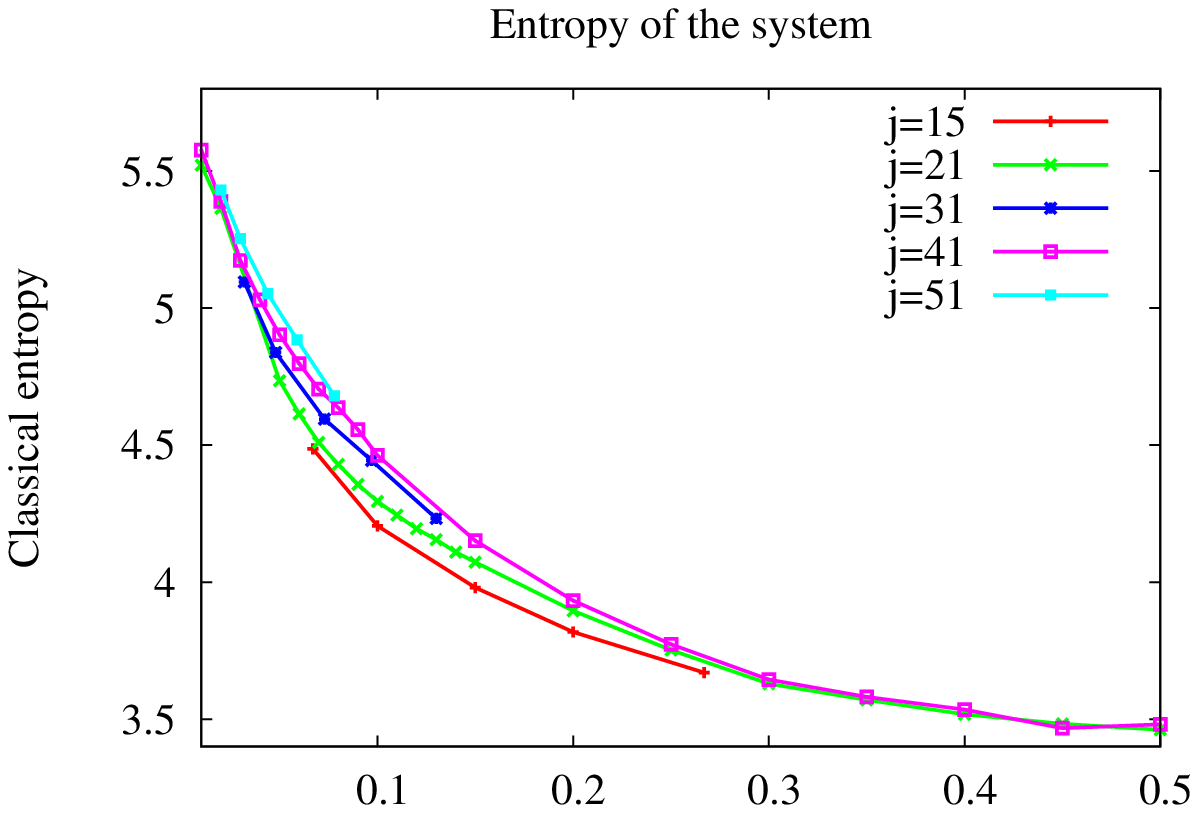}} {\includegraphics[scale=0.66]{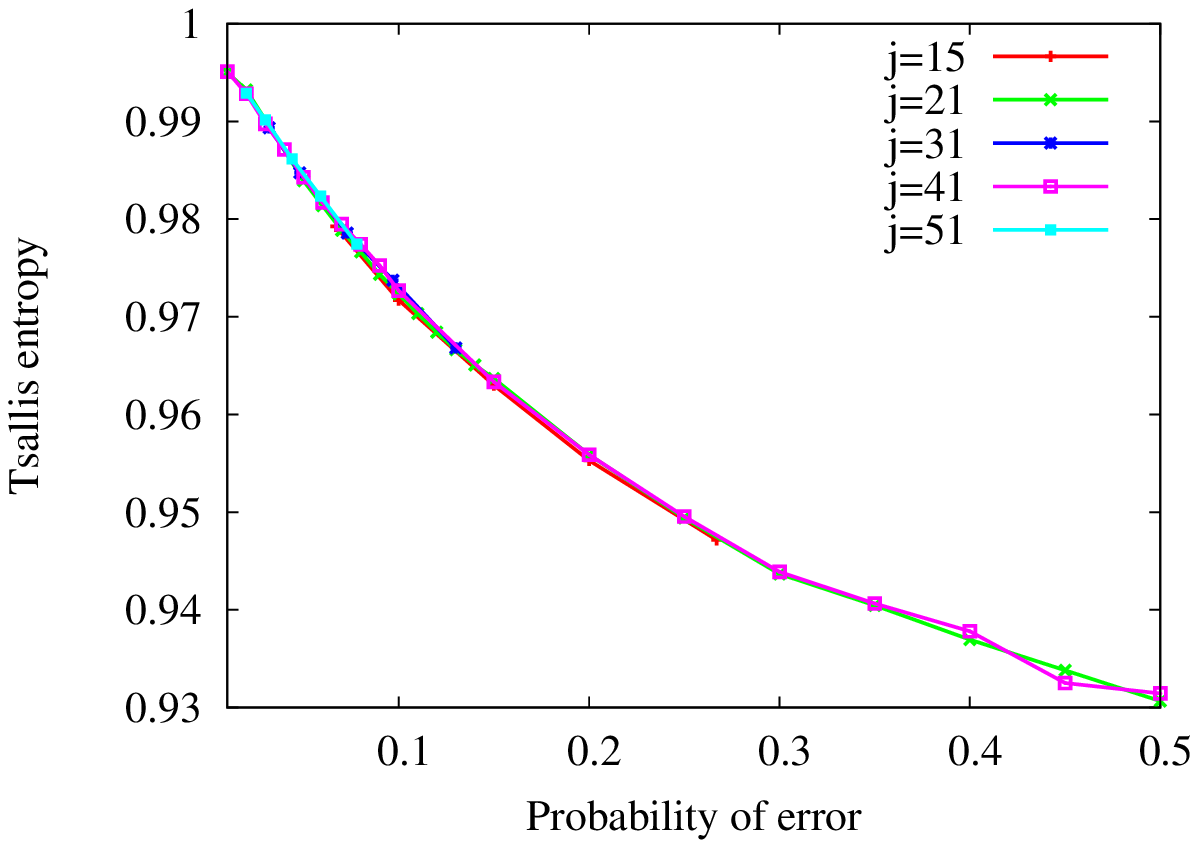}} 
\par\end{centering}

\caption{Classical entropy of the probability distribution is on the top, Tsallis
entropy for $q=2$ on the bottom for parameters
set to $T=200$ and $R=15000$. Classical entropy, on the top, decreases
for increasing $p$ but the shape of the dependence is different for different values of
$j$. The shape of the decrease becomes the same for Tsallis entropy with
parameter $q=2$ as seen on the bottom.
}

\label{Quantum_random_walk_odd_entropy} 
\end{figure}

Finally, the last studied variable for the QW with static disorder
was the variance of position displacements $Var\ \mathbb{P}_{p,j}$,
plotted in Fig.~\ref{Qunatum_random_walk_odd_statistics}
on the left and with rescaled axes on right-hand side. The variance forms 
a U-shaped function of $p$ with moving minima for different but fixed $j$.
Rescaling the axis $x=p\cdot j^{\alpha}$ where $\alpha=1.04$ and axis
$y=j^{-\beta}\cdot Var\ \mathbb{P}_{p,j}$ where $\beta=1.67$,
we clearly see from Fig.~\ref{Qunatum_random_walk_odd_statistics}, right 
inset, that there is a universal U-shape function $f_{\mathbf{odd}}^{*}$(x)
fulfilling
\begin{equation}
f_{\mathbf{odd}}^{*}\left(x\right)=j^{-\beta}\cdot Var\ 
\mathbb{P}_{x,j}.
\end{equation}
 Due to relation $\alpha\simeq1$ and the U-shape of function $f_{\mathbf{odd}}^{*}\left(x\right)$
we can conclude that the overall variance of probability distribution
of the walker's position is strongly correlated with the fit of the
inverse parameter $1/a$ (of the Laplace distribution) of whole
probability distribution.

\begin{figure}
\begin{centering}
\includegraphics[scale=0.67]{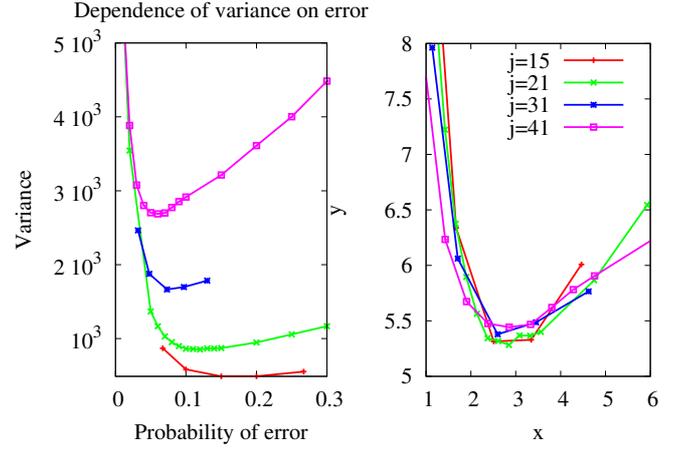}

\par\end{centering}

\caption{Dependence of variance of the probability distribution for odd
$j$ on $p$
for $T=200$ and $R=15000$ is shown. The typical U-shape function
can be observed with shifted position of the minimum. On the right, there
are the same data with rescaled $x$-axis where $x=p\cdot j^{\alpha}$
where $\alpha=1.04$ and $y$-axis where $y=j^{-\beta}Var(i)$ where
$\beta=1.67$. The data fit on the same U-shape function reaching their minimum at the same place.}

\label{Qunatum_random_walk_odd_statistics} 
\end{figure}

\subsubsection{Even jumps}

Let us now assume that the ensemble of the random unitary matrices is still 
parametrized by the probability $p$ that the pair of erroneous positions 
occur, but the distance of a pair of errors $j$ is an even number and due to the
mappings \ref{Mapping_of_states_1} and \ref{Mapping_of_states_2} the chirality 
of walker does not change. 
The typical formation of Laplace-like distribution as in the case
of odd jumps is not present, but instead we observe a $3$-peaked structure, as
seen in Fig.~\ref{Quantum_random_walk_even_jumps} (left-right asymmetric
due to the initial condition) modulated by a periodic function that has its
period equal to the length of jump $j$. To emphasize the difference between
odd and even $j$, we conclude that both cases form a located peak
at the initial position of the walker. However, odd sized jumps $j$ cause
Laplace-like overall distribution while the case of even $j$ can form
a $3$-peaked distribution where the central (localized) peak is
sharp but the two others are broad and peaks with distance $j$ are modulated
on the overall structure.

\begin{figure}

\begin{centering}
{\includegraphics[scale=0.66]{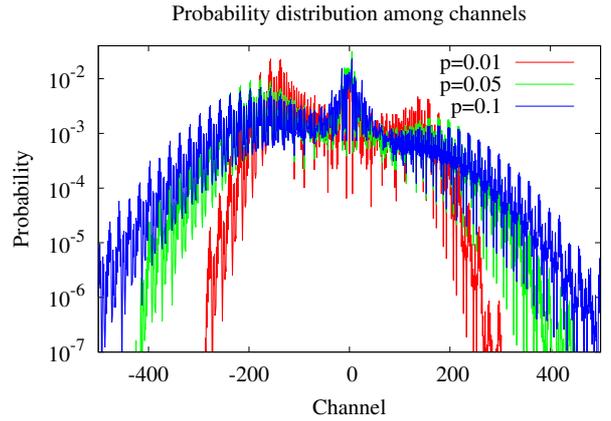}} 
\par\end{centering}

\caption{Probability distribution among channels for $T=200$ and $R=15000$
in the case of even sized jumps ($j=20$). The probability distribution
forms a $3$-peak structure that is modulated by additional periodical small peaks
which are separated by $j$.}

\label{Quantum_random_walk_even_jumps} 
\end{figure}

\subsection{Dynamic disorder}

The static disorder analyzed in the preceding sections assumes random, but 
time-correlated emergence of pairs of errors that are expressed in the set of 
random unitary operators $S(\Omega,{\cal F},\mathbb{P})$. In contrast, 
dynamic disorder assumes independent and identically distributed random
unitary operators from the same probability space $S(\Omega,{\cal F},\mathbb{P})$.

Let us consider on the probability distribution of the position of the walker, plotted in
Fig.~\ref{Quantum_random_walk_bulky_pattern} on the top. We observe
decoherence of the quantum walker forming  a distribution reminiscent of
a Gaussian distribution modulated by $1$D QW patterns for small probabilities
of jump $p$.  On the other hand, large $p$ causes a modulation by valleys
with distance $j$ between peaks. The functional dependence of the standard
deviation of the probability distribution of position of the walker on
probability of jump $p$ in Fig.~\ref{Quantum_random_walk_bulky_pattern}
on the bottom shows a linear behavior.

\begin{figure}
\begin{centering}
{\includegraphics[scale=0.66]{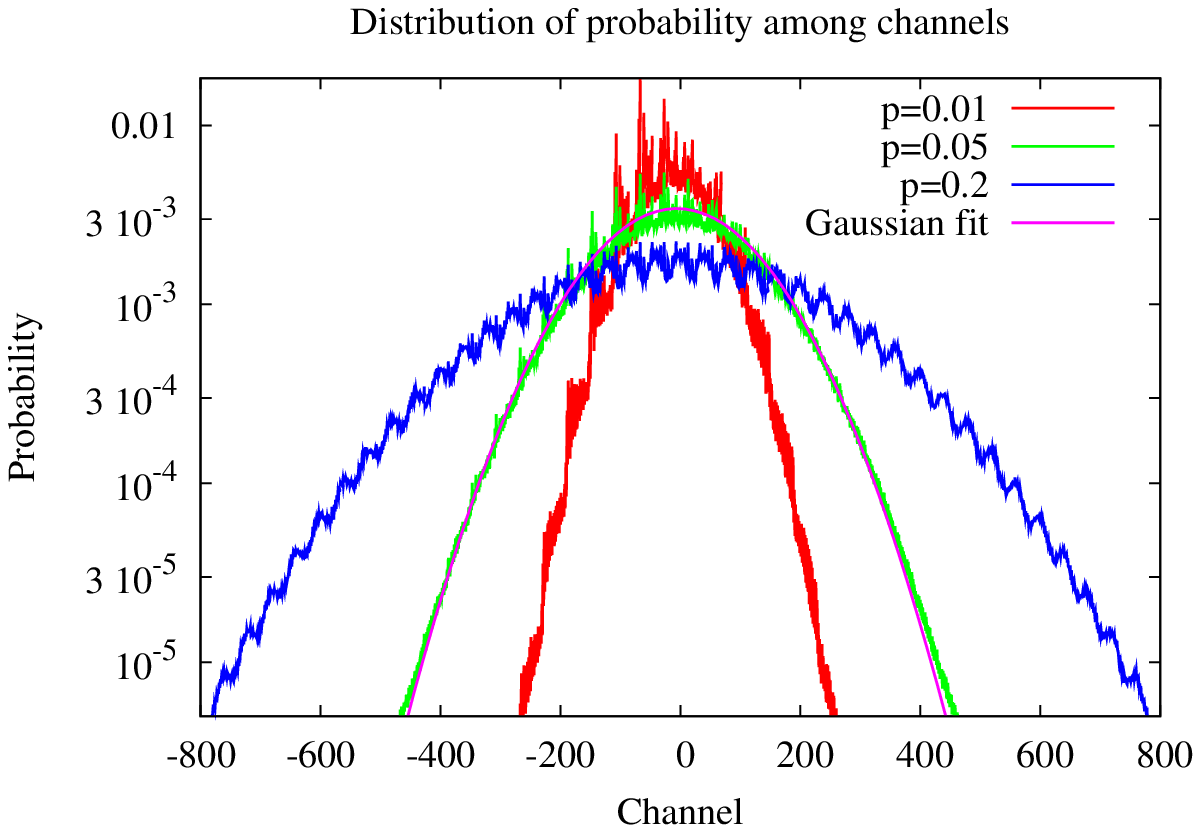}}
{\includegraphics[scale=0.66]{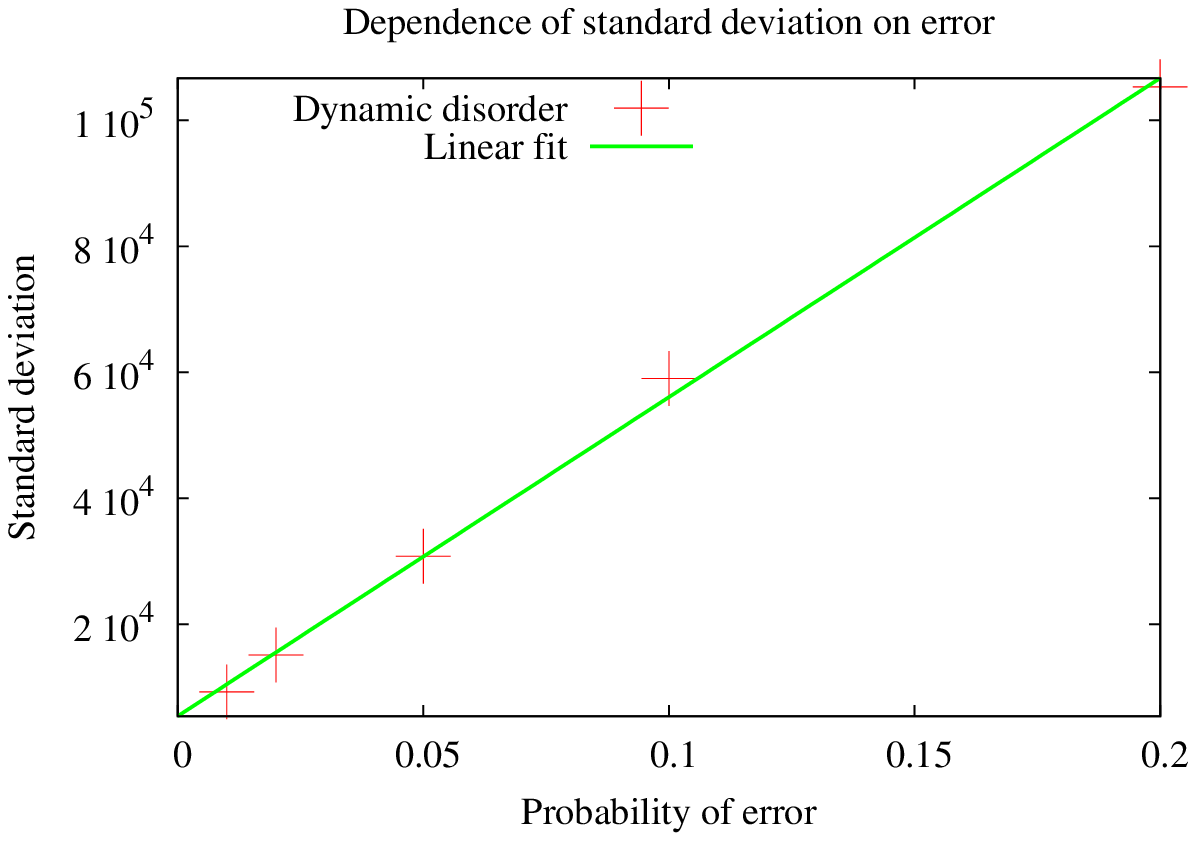}} 
\par\end{centering}

\caption{On top, probability distribution among the channels for dynamic disorder
for $T=100$, $R=5000$ and $j=40$. Probability distribution can
be fitted to a normal distribution where residual patterns of the QW
are present for $p=0.01$; for $p=0.2$, formation of periodical
valleys modulated on overall distribution (similar to static disorder
with an odd size of jump) can be observed. On the bottom, functional 
dependence of standard deviation of the probability distribution of the 
position of the walker on probability of error $p$ and its linear fit.}

\label{Quantum_random_walk_bulky_pattern} 
\end{figure}

{}

\section{Discussion and Conclusions}

We have defined a modification of the QW in $1$D where the environment causes long 
but fixed-size jumps that emerge with a constant probability---the model is no 
longer deterministic but stochastic, depending on a random variable.
In one step, first a unitary QW coin and shift operator act on the state, then a stochastic displacement operator generates jumps.
We have studied two classes of QW with disordered 
connections between beam-splitters. First, dynamic disorder model shows 
decoherence that leads to Gaussian distribution modulated by residual patterns 
of QW (for small probabilities of jump $p$) or by valleys (for large 
probabilities of jump $p$). The standard deviation of the position of the 
walker is a linear function of $p$. In the second case, the QW is perturbed by 
static disorder. The behaviour of the model depends on the size of jump and we have 
investigated odd and even jump sizes separately.
Even jump sizes cause localization of the walker at the initial position and 
two other broad peaks modulated by oscillations with a period of the size of 
jump $j$. The focus of the paper, however,  was mainly given to jumps with odd 
size.  In this case, the evolution of the variance shows a transition from the 
ballistic diffusion to a no-diffusion regime with increasing $p$ and this 
observation is supported by numerical calculation of classical entropy which 
changes from logarithmically increasing regime to a no-growth regime.

In addition, the probability distribution of positions of the walker changes from 
a pattern typical to the QW to Laplace distribution modulated by Laplace 
distributed peaks separated by the size of jump $j$. The formation of Laplace 
distribution  depends on the probability $p$ that a pair of errors occurs. 
Moreover, the model exhibits the unusual property that classical disorder in 
quantum system can decrease, i.e., we observe a decrease of both the classical 
(Shannon) and $q$-entropy (with $q=2$) of measurements.  Finally, our numerical calculation shows that using non-additive 
$q$-entropy with $q=2$ there is an universal functional dependence of 
$q$-entropy on $p$ for arbitrary $j$. 
The next part of out investigation was turned to the variance of the 
probability distribution. Our numerical results indicate that there is an 
universal functional form of variance of the probability distribution of the 
position of the walker---variance multiplied by $j^{-\beta}$ is an universal 
U-shaped function of one variable $p\cdot j^{\alpha}$ where $\alpha=1.04$ and 
$\beta=1.67$. The functional dependence of the universal function shows a minimum 
which separates two regimes---one with decreasing variance and the second with 
increasing variance. To put this result in the broader context of Complex Systems 
and Game Theory, we note that similar behavior of variance of attendance has been observed 
in Minority Game, exhibiting dynamical phase transition, see, e.g., 
\cite{Lavicka-2010,challet-1998}.

\section{Acknowledgements}

The financial support by MSM 6840770039, M\v SMT LC 06002, GA\v CR 202/08/H072, the Czech-Hungarian cooperation project (KONTAKT CZ-11/2009), Hungarian Scientific Research Fund (OTKA) under Contract No. K83858, the Emmy Noether Program of the DFG (contract No LU1382/1-1) and the cluster of excellence Nanosystems Initiative Munich (NIM) is gratefully 
acknowledged.

\appendix

\section{Properties of the set of operators}

Let the unperturbed walking space be represented by a cycle graph with $N$ 
vertices where every vertex represents a channel in the main text.  The errors in the network are represented by swapping the walker's 
probability amplitudes between vertices labelled $i$ and $i+j\bmod N$ for 
a fixed $j$, which happens with a relative probability $p$. With a relative 
probability $1-p$, a vertex is left intact. Finally, a vertex already 
exchanged with another one cannot be used for another transposition in the 
same permutation.

These conditions give the probability of a permutation $\pi$ in the form
\begin{equation}
\mathbb{P} (\pi) = \frac{1}{Z^{(N)}} p^{tr(\pi)} (1-p)^{N - 2 \cdot tr(\pi)},
\end{equation}
where $tr(\pi)$ denotes the unique number of independent transpositions 
forming the permutation $\pi$. A factor of $2$ in the exponent of $1-p$ is 
present due to the fact that every transposition reduces the number of unused 
vertices by $2$. Finally, the factor $Z^{(N)}$ is the normalization constant.

The constant $Z^{\left( N \right) }$ is computed as
\begin{equation}
Z^{\left( N \right) }= \sum_{k=0}^{\lfloor \frac{N}{2} \rfloor} N_k p^k \left( 
1- p \right)^{N-2\cdot k },
\end{equation}
where $N_k$ is the count of all possible permutations formed by exactly $k$ 
non-incident transpositions of size $j$.

\begin{lemma}
\label{Lemma 1}
Let $N>2$, let $j$ and $N$ are relatively prime. Then 
\begin{equation}\label{Equation for Z}
Z^{\left( N \right) }=1+\left( -p \right)^N .
\end{equation}
\end{lemma}

\begin{proof}
Due to the relative primality of $j$ and $N$, we can relabel the vertices by 
indices $0$ through $N$ such that vertices with successive indices have 
a distance of $j$ in the original numbering. This way, we can reduce the 
problem of finding $N_k$ to a combinatorial problem of finding the number of 
$k$-element subsets $A \subseteq \lbrace 0, 1, \ldots, N-1\rbrace$ satisfying 
the following conditions:
\begin{enumerate}[(a)]
\item for all $0 \le i < N-1$, $\lbrace i, i+1 \rbrace \nsubseteq A$,
\item $\left\lbrace 0, N-1 \right\rbrace \nsubseteq A$.
\end{enumerate}

In order to find $N_k$, we discuss two disjoint cases:
\begin{enumerate}[(i)]
\item Let us count the subsets which do not contain $N-1$ as an element. For 
each such set $A = \lbrace a_1, a_2, a_3, \ldots,$ $a_k\rbrace$, $a_1 < a_2 
< \ldots < a_k$, we denote $\tilde A = \lbrace a_1, a_2-1, a_3-2, \ldots, 
a_k-k+1 \rbrace$. This is a one-to-one mapping, reducing the problem to 
finding $k$-element subsets of $N-k$ elements without any additional 
restriction. This gives $\binom{N-k}{k}$ possible subsets.
\item Let now $N-1 \in A$. Then the other $k-1$ elements of $A$ must lie in 
$\lbrace 1, 2, \ldots, N-2 \rbrace$, with no two of them successive and $N-2$ 
excluded. This is a variant of the subproblem (i), giving $\binom{N-k-1}{k-1}$ 
possibilities for $A$.
\end{enumerate}
Adding these results, we find that
\begin{equation} \label{N_k}
N_{k} = \binom{N-k}{k} + \binom{N-k-1}{k-1}.
\end{equation}

In order to calculate $Z^{(N)}$, we find the generating function
\begin{eqnarray}
F(x) &=& \displaystyle\sum_{N=0}^{+\infty} Z^{(N)} x^N 
 \nonumber \\
&=&
 \displaystyle\sum_{N=2}^{+\infty} \sum_{k=0}^{\lfloor \frac{N}{2} \rfloor}
N_k
 \times p^k (1-p)^{N-2\cdot k} x^N.
\end{eqnarray}
Here, for simplicity, we generalize \ref{N_k} also for $N \le 2$ and we define 
by convention $\binom{-1}{-1} = 0$. Using standard methods, we obtain for the 
sum
\begin{equation}
F(x) = \frac{1+px^2}{1-x+xp-px^2}.
\end{equation}
This result can be decomposed into partial fractions as
\begin{equation}
F(x) = \frac{1}{1-x} +\frac{1}{1+px} -1,
\end{equation}
from which the power series can be derived quickly as
\begin{eqnarray}
F(x) = \sum_{N=0}^{+\infty} x^N + \sum_{N=0}^{+\infty} \left(-px\right) ^N - 1 \nonumber \\
= -1 + \sum_{N=0}^{+\infty} \underbrace{ \left( 1 + \left( -p  \right)^N \right) }_{Z^{\left(N \right) }}  x^N.
\end{eqnarray}
QED.
\end{proof}

If $N=2$, the formula \ref{Equation for Z} cannot be used.  Indeed, computing 
$Z^{(2)}$ manually gives
\begin{equation}
Z^{(2)} = p + (1-p)^2 = 1 - p + p^2.
\end{equation}
This is because \ref{N_k} gives an incorrect result for $k=1$ in this case.  
In practical situations, however, $N \gg 2$.

\begin{lemma}
Let $N$ and $j$ be positive integers such that $j < N$, $g 
= \mathop{\mathrm{gcd}}(N,j)$ and $\frac{N}{g}>2$. Then
\begin{equation}\label{Equation for Z2}
Z^{\left( N \right) } = \left(1+\left( -p \right)^{\frac{N}{g}} \right)^g.
\end{equation}
\end{lemma}

\begin{proof}
If $g = 1$, then $j$ and $N$ are relatively prime and we can use 
Lemma~\ref{Lemma 1} to obtain the result. In the case ${g} > 1$, we 
first split the set $\lbrace 0, 1, 2, \ldots, N-1 \rbrace$ into $g$ modular 
classes $\pmod g$.  According to the definition, the selection of the errors 
can be done independently on each of these subsets. Therefore, as 
$\frac{N}{g}$ is relatively prime to $j$, we can use Lemma~\ref{Lemma 1} for 
each of these classes and multiply the partial results to obtain the total 
partition function in the form stated by the Lemma.
\end{proof}

\bibliographystyle{unsrt}
\bibliography{bibliography}{}
\end{document}